\documentclass[AMS,STIX1COL]{WileyNJD-v2}

\articletype{RESEARCH ARTICLE}%

\usepackage{amsmath,amssymb,amsfonts}
\usepackage{graphicx}
\usepackage{rotating}
\usepackage{float}
\usepackage{color}
\usepackage{subfigure}
\begin{document}

\title{Adaptive tracking control for non-periodic reference signals under quantized observations}

\author[1]{Chuiliu Kong}

\author[2]{Ying Wang}

%\author[2]{Yanlong Zhao*}

\authormark{KONG \textsc{et al}}

\address[1]{\orgdiv{Research Center for Mathematics and Interdisciplinary Sciences}, \orgname{Shandong University}, \orgaddress{\state{Qingdao, Shandong Province, 266237}, \country{China}}}

\address[2]{\orgdiv{Key Laboratory of Systems and Control, Institute of Systems Science, Academy of Mathematics and Systems Science}, \orgname{Chinese Academy of Sciences}, \orgaddress{\state{Beijing 100190}, \country{China}}}

%\address[3]{\orgdiv{Key Laboratory of Systems and Control, Institute of Systems Science, Academy of Mathematics and Systems Science}, \orgname{Chinese Academy of Sciences}, \orgaddress{\state{Beijing 100190}, \country{China}}, and also with \orgdiv{School of Mathematics Sciences}, \orgname{University of Chinese Academy of Sciences}, \orgaddress{\state{Beijing 100049}, \country{China}}}

\corres{*Ying Wang,  \email{wangying96@amss.ac.cn}}

%\presentaddress{This is sample for present address text this is sample for present address text}

\abstract[Summary]{	This paper considers an adaptive tracking control problem for stochastic regression systems with multi-threshold quantized observations. Different from the existing studies for periodic reference signals, the reference signal in this paper is non-periodic. Its main difficulty is how to ensure that the designed controller satisfies the uniformly bounded and excitation conditions that guarantee the convergence of the estimation in the controller under non-periodic signal conditions. This paper designs two backward-shifted polynomials with time-varying parameters and a special projection structure, which break through periodic limitations and establish the convergence and tracking properties. To be specific, the adaptive tracking control law can achieve asymptotically optimal tracking for the non-periodic reference signal; Besides, the proposed estimation algorithm is proved to converge to the true values in almost sure and mean square sense, and the convergence speed can reach $O\left(\frac{1}{k}\right)$ under suitable conditions. Finally, the effectiveness of the proposed adaptive tracking control scheme is verified through a simulation.}


%\JELinfo{classification}

%\MSC{Code numbers}

%\jnlcitation{\cname{%
%\author{C. Kong} and
%\author{Y. Wang}},
%\ctitle{Adaptive tracking control for non-periodic reference signals under quantized observations}, \cjournal{Int J Robust Nonlinear Control}, \cvol{xx;xx:x--xx}.}

\maketitle

%\footnotetext{\textbf{Abbreviations:} ANA, anti-nuclear antibodies; APC, antigen-presenting cells; IRF, interferon regulatory factor}

\section{Introduction}\label{sec1}

Since the actual control systems usually suffer from uncertainties stemming from parameters, structure, and environment (e.g., system aging, component failures, and external disturbances), the development of adaptive control theory provides a method to deal with the above uncertainties. In addition, owing to the limited communication ability and computation ability of the sensors, the outputs of the actual control systems are often with finite precision, i.e., it is only known whether the output belongs to a set, not the exact value (e.g., \cite{Wang-Kim-Sun},  \cite{Wang-Zhang-Yin} and \cite{Wang-Hu-Zhao-D}). Therefore, considering the actual problem requirements, the quantized adaptive control techniques provide an effective approach to control problems involving quantized measurements.

Usually, the design of adaptive control is based on the estimation of unknown system parameters to obtain the desired system performance. For the difference of the unknown parameters to be estimated in the controller, the design of adaptive control is subdivided into direct adaptive method and indirect adaptive method. In detail, in the direct adaptive method (e.g., \cite{Wen-Tao}, \cite{Zhang-Tao} and \cite{Zhang-Zhang}), we directly estimate the unknown parameters included in the designed controller. Whereas in the indirect adaptive method (e.g., \cite{Chen-Guo-1}-\cite{Guo-2},  and \cite{Li}-\cite{Wang-Lai}, \cite{Zhang-Guo}, \cite{Chen-Guo}), the original unknown parameters in the system are directly estimated and then the adaptive control is designed by the above estimations. This means that the parameters required in the former method do not necessarily have to estimate the true value of the unknown parameter in the system, whereas that true value is required in the latter method. For example, in the framework of the direct adaptive control, Zhang et al. \cite{Zhang-Tao,Zhang-Zhang} proposed an implicit function-based adaptive control scheme for the discrete-time neural-network systems and nonlinear systems in a general non-canonical form, respectively. Wen et al. \cite{Wen-Tao} studied the convergence of the output tracking error for nonlinear adaptive control systems. It should be noted that they consider deterministic systems. Moreover, in the framework of the indirect adaptive, there are many classic methodologies about stochastic adaptive control such as stochastic gradient algorithm-based adaptive control method \cite{Chen-Guo-1,Peter-Lafortune} and least-squares algorithm-based self-tuning regulator \cite{Guo-Chen,Guo-1}, and so on. Numerous adaptive control theories have been derived from the aforementioned methods and techniques, such as \cite{Guo, Guo-2, Zhang-Guo, Chen-Guo}. In this paper, the tracking problem is studied in the indirect adaptive framework. The reason is that on the one hand, the original unknown parameters in the system can be estimated accurately, on the other hand, the designed adaptive control based on the estimations of the unknown parameters is more robust.

It is worth noting that the aforementioned methodologies are constructed based on the accurate information of system output, or output with measurement noises. However, as mentioned above, obtaining the exact outputs of the systems may be not achievable, and in some cases only quantized measurements of the system outputs can be obtained. These quantized measurements inherently contain measurement uncertainty about the system outputs. Therefore, the research goal of this paper is how to use the quantized observations to design an adaptive control law to achieve accurate control of the systems, which falls into the category of robust control. The existing research on adaptive control with quantized measurements has been somewhat studied but is still relatively lacking. For example, for first-order stochastic regression systems with binary-valued observations, Guo et al. in \cite{Guo-Zhang-Zhao} designed an adaptive tracking control and finally achieved the asymptotically optimal tracking for deterministic reference signal sequence. Further, for the high-order stochastic regression systems with binary-valued observations, Zhao et al. in \cite{Zhao-Guo-Zhang} proposed a two-scale algorithm and achieved the asymptotically optimal tracking for the periodic reference signal. In this algorithm, on the one hand, the designed adaptive control is required to last at a long time scale. On the other hand, the unknown parameter is estimated by the empirical measurement method \cite{Wang-Zhang-Yin} at a small time scale, which is an offline identification algorithm. Since adaptive control in the above two-scale algorithm is limited to remain constant over time, it results in a possible reduction of the tracking speed. Therefore, to further optimize the tracking speed, with the help of an online recursive projection identification algorithm in \cite{Guo-Zhao-1}, Wang et al. in \cite{Wang-Hu-Zhao} designed another adaptive tracking control law without waiting time and finally achieved the desired tracking effect for the periodic reference signal. For the case of system outputs that are quantized and also contain observation uncertainty, Kong et al. in \cite{Kong} investigated an adaptive tracking control problem for periodic reference signals and obtained the effect of observation uncertainty on the designed algorithm. Furthermore, a detailed description of the quantized system identification theory can be referred to \cite{Wang-Yin-Zhang-Zhao}.

At present, by summarizing the existing research as mentioned above, it is easy to see that there are two kinds of shortcomings in the tracking control problem for stochastic regression systems containing quantized observations: one is that the system model studied is relatively simple, that is, the system model is first-order \cite{Guo-Zhang-Zhao}, the other is that the reference signal is limited to periodicity \cite{Kong, Wang-Hu-Zhao, Zhao-Guo-Zhang}. The reason for considering the quantized tracking control problem in the above two types of frameworks is to facilitate the solution of the following two problems, which are precisely the ones that need to be overcome in the framework of this paper, i.e.,

$\bullet$ The adaptive tracking control design is coupled with the parameter estimation algorithm \cite{Guo-2}; 

$\bullet$ Losing the periodicity condition, in contrast to  \cite{Kong, Wang-Hu-Zhao, Zhao-Guo-Zhang}, this paper faces difficulties in ensuring that controllers designed in accordance with the certainty equivalent principle satisfy both persistent excitation and uniformly bounded conditions.

Therefore, this paper is concerned with the quantized tracking control problem where the stochastic regression system is two-order and the reference signal is non-periodic at the same time, whose model is relatively more general than the first-order quantized system in \cite{Guo-Zhang-Zhao} and reference signals are more general than the periodic case in \cite{Kong, Wang-Hu-Zhao, Zhao-Guo-Zhang}. This paper designs two backward-shifted polynomials with time-varying parameters and a special projection structure to overcome the aforementioned two difficulties. However, it also should be noted that this control scheme cannot be directly used to handle the quantized system with higher-order cases. It is because, for higher-order quantized systems, there remains some difficulty in ensuring that adaptive tracking control law designed by the certainty equivalence principle satisfies the uniformly bounded condition. Actually, this paper is the first attempt to address the tracking problem for non-periodic reference signals under higher-order quantized systems, and provides referential ideas for further solving the tracking control problem under more general quantized system models. The main contributions of this paper can be summarized as follows.
%The reason for studying two-order quantized systems is that compared to the first-order system case \cite{Guo-Zhang-Zhao}, the quantized system model considered in this paper is relatively general and the proposed adaptive tracking control scheme for this model can handle non-periodic reference signals.
\vspace{4pt}

i) Model and Problem aspects: On the one hand, the outputs of the stochastic regression model considered in this paper are quantized by a multi-threshold quantizer, which introduces a strong non-linearity to the system. On the other hand, the adaptive tracking control problem considered in this paper is for the non-periodic reference signals. This research breaks through the limitations of previous literature on first-order system model \cite{Guo-Zhang-Zhao} and periodic reference signals (e.g., \cite{Kong}, \cite{Wang-Hu-Zhao} and \cite{Zhao-Guo-Zhang}). It makes the system models and the research results studied in this paper to be extended to a wider scope.

\vspace{4pt}

ii) Method aspect: In this paper, we design two backward-shifted polynomials with time-varying parameters and a special projection structure to make the designed adaptive tracking control law satisfy the persistent excitation and uniformly bounded conditions, which are used to guarantee the convergence properties of the estimation algorithm. Compared with \cite{Guo}, the system outputs considered in this paper are quantized by a multi-threshold quantizer.

\vspace{4pt}

iii) Result aspect: The adaptive tracking control scheme designed in this paper not only ensures that the estimations of the unknown parameters converge to the true values but most importantly also achieves the asymptotically optimal tracking effect for the non-periodic reference signals. In addition, the mean square convergence speed for the unknown parameter estimation can also reach $O\left(\frac{1}{k}\right)$ as that in \cite{Wang-Hu-Zhao}.

\vspace{4pt}

The rest of this article is organized as follows. In Section 2, we introduce the tracking control problem with multi-threshold quantized outputs. In Section 3, an adaptive tracking control scheme is presented. In Section 4, the main results including the convergence, convergence speed, and the performance of adaptive tracking control are given. In Section 5, a simulation is given to verify the above main conclusions. Finally, we summarize the conclusions of this article in Section 6.

\textbf{Notations.}
First of all, we define some notations. Let $\|\cdot\|$ be the Euclidean norm, $\lfloor x\rfloor =\mbox{max}\{a\in\mathbb{Z}|a\leq x\}$ and $\lceil x\rceil=\mbox{min}\{a\in\mathbb{Z}|a\geq x\}$ be the largest integer less than or equal to $x$ and the smallest integer large than or equal to $x$ for any $ x\in\mathbb{R}$, respectively. For two real number sequence $\{a_{1,k},\ k\geq 1\}$ and $\{a_{2,k},\ k\geq 1\}$, one writes $a_{1,k}=O(a_{2,k})$ if and only if there exist a positive real number $M_0$ and an integer $k_0$ such that $|a_{1,k}|\leq M_0|a_{2,k}|$ for all $k\geq k_0$.
%%%%%%%%%%%%%%%%%%%%%%%%%%%%%%%%%%%%
\section{Problem formulation}
In this section, we consider the following stochastic regression model:
\begin{equation}\label{multi-system}
	y(k)=\phi^{\intercal}(k)\theta+w(k),\ k=1,2,...
\end{equation}
where $\phi(k)\stackrel{\vartriangle}{=}[u(k),u(k-1)]^{\intercal}\in\mathbb{R}^{2}$, $u(s)=0$, for $s\leq0$, $\theta\stackrel{\vartriangle}{=}[\theta(1),\theta(2)]^{\intercal}\in\mathbb{R}^{2}$, and $\{w(k)\}_{k\geq1}$ are the control input, unknown parameter and system noise, respectively. $y(k)\in\mathbb{R}$ is the measured output, which can not be exactly measured and only its quantized observation can be obtained by a sensor with multi-threshold $C_p$, for $p=0,1,...,m$, and $C_0<C_1<C_2<\cdots<C_m<C_{m+1}$. Then, the quantized observation is
\vspace{-10pt}
\begin{equation}\label{q-information}
	{S}(k)=\sum_{p=0}^{m}pI_{\{C_p< y(k)\leq C_{p+1}\}},
\end{equation}
\vspace{-1pt}
where $C_0=-\infty,\ C_{m+1}=+\infty$ and the indicator function $I_{\{C_p< y(k)\leq C_{p+1}\}}$ is defined by
\begin{equation*}
	I_{\{C_p< y(k)\leq C_{p+1}\}}=
	\left\{
	\begin{aligned}
		&1,\quad\mbox{if}\ y(k)\in(C_p,C_{p+1}],\\
		&0,\quad\mbox{otherwise}.
	\end{aligned}
	\right.
\end{equation*}
\vspace{-1pt}
Next, some necessary assumptions are introduced. Firstly, set $\{y^{*}(k)\}_{k\geq1}$ as the known reference signal, which satisfies the following two assumptions.

\begin{assumption}\label{ass-reference-signal}
	The reference signal $\{y^{*}(k)\}_{k\geq1}$ is a bounded sequence, i.e., $|y^{*}(k)|\leq\bar{y}^{*}$, where $\bar{y}^{*}$ is a given positive constant.
\end{assumption}
\begin{assumption}\label{ass-multi-reference-signal}
	Denote $\mathcal{Y}(i){=}[y^{*}(i),y^{*}(i-1)]^{\intercal}$. There exist constants $h> 2$ and $\delta_y>0$ such that
	\begin{equation*}
		\lambda_{\min}\left(\sum_{i=k+1}^{k+h-1}\mathcal{Y}(i)\mathcal{Y}^{\intercal}(i)\right)>\delta_y,\ k=1,2,...
	\end{equation*}
\end{assumption}

\begin{remark}
	The purpose of this paper is to design an adaptive tracking control $\{u(k)\}_{k\geq 1}$ to drive the output $\{y(k)\}_{k\geq1}$ to follow the non-periodic reference signal $\{y^{*}(k)\}_{k\geq1}$. Therefore, the model and problem considered in this paper are fundamentally different from the previous literature on first-order system model \cite{Guo-Zhang-Zhao} and periodic reference signals (e.g., \cite{Kong}, \cite{Wang-Hu-Zhao} and \cite{Zhao-Guo-Zhang}). 
\end{remark}

Secondly, in order to proceed with our analysis, some other assumptions about the unknown parameter and noises are introduced as follows.
\begin{assumption}\label{ass-minimum-phase-condition}
	There exists a constant $\mu\in(0,1)$ such that $|Q(\omega)|>\mu$ for $|\omega|\leq 1$ with $Q(\omega)=\theta(1)+\theta(2)\omega$. 
\end{assumption}

\begin{remark}\label{re-1}
	If $\mu=0$, then Assumption \ref{ass-minimum-phase-condition} turns to be the minimum phase condition, which is necessary for adaptive
	tracking problem. Moreover, Assumption \ref{ass-minimum-phase-condition} implies that the unknown parameter $\theta(1)\neq0$ and $\left|\frac{\theta(1)}{\theta(2)}\right|>1$.
\end{remark}

\begin{assumption}\label{ass-multi-bounded}
	Let $\Theta$ denote the set of parameters satisfy Assumption \ref{ass-minimum-phase-condition}. The parameter $\theta$ belongs to a convex compact set $\Omega$ and $\Omega$ is the subset of $\Theta$, i.e., $\theta\in\Omega\subseteq\Theta$. Besides, $\bar{M}=\sup_{v\in\Omega}\|v\|$, $0<\underline{\theta}\leq|\theta(1)|$ and $|\theta(2)|<\bar{\theta}$, $\underline{\theta}$ and $\bar{\theta}$ are constants and $\underline{\theta}>\bar{\theta}>0$.
\end{assumption}

\begin{remark}\label{re-2}
	Let us give an example of Assumptions \ref{ass-minimum-phase-condition}-\ref{ass-multi-bounded}. For example, let $\underline{\omega}\stackrel{\vartriangle}{=}\underline{\theta}I_{\{sgn(\theta(1))>0\}}-\bar{M}I_{\{sgn(\theta(1))<0\}}$ and $\bar{\omega}\stackrel{\vartriangle}{=}\bar{M}I_{\{sgn(\theta(1))>0\}}-\underline{\theta}I_{\{sgn(\theta(1))<0\}}$, if the sign of the unknown parameter $\theta(1)$ is known, then the set $\Omega$ in Assumption \ref{ass-multi-bounded} can be constructed as $\{(\omega(1),\omega(2))|\underline{\omega}\leq\omega(1)\leq \bar{\omega},\ |\omega(2)|\leq \bar{\theta}\}$, where the positive constants $\bar{\theta},\ \underline{\theta},\ \bar{M}$ satisfy $\bar{\theta}<\underline{\theta}\leq \bar{M}$. 
	
	%the set $\Omega$ in Assumption \ref{ass-multi-bounded} can be constructed as $\{(\omega(1),\omega(2))|\underline{\omega}\leq\omega(1)\leq \bar{\omega},\ |\omega(2)|\leq \bar{\theta}_2\}$, where if the sign of the unknown parameter $\theta(1)$ is known, then $\underline{\omega}\stackrel{\vartriangle}{=}\underline{\theta}I_{\{sgn(\theta(1))>0\}}-M_1I_{\{sgn(\theta(1))<0\}}$, $\bar{\omega}\stackrel{\vartriangle}{=}M_1I_{\{sgn(\theta(1))>0\}}-\underline{\theta}I_{\{sgn(\theta(1))<0\}}$ and the positive constants $\hat{\theta}_2,\ \underline{\theta},\ M_1$ satisfy $\hat{\theta}_2<\underline{\theta}\leq M_1$. 
\end{remark}

\begin{assumption}\label{ass-noise}
	$\{w(k),\ k\geq 1\}$ is an i.i.d. random sequence with zero mean and finite variance. The probability distribution function $F(\cdot)$ of $w(1)$ is assumed to be known. 
\end{assumption}

Then, based on the above system \eqref{multi-system}-\eqref{q-information}, the tracking control problem considered in this paper is detailed as follows:
\begin{problem}\label{multi-adaptive-control-problem}
	Let $\{y^{*}(k)\}_{k\geq1}$ be a known non-periodic reference signal and satisfy Assumptions \ref{ass-reference-signal} and \ref{ass-multi-reference-signal}. The aim of this paper is to construct an adaptive control $\phi(k)$ based on the quantized observations $\{{S}(1),{S}(2),...,{S}(k-1)\}$ and the past inputs $\{u(1),u(2),...,u(k-1)\}$ to minimize the following index:
	\begin{equation}\label{multi-index}
		J(k)=\mathbb{E}\left[(y(k)-y^{*}(k))^2\right],\ \forall k\geq 1.
	\end{equation}
\end{problem}
%%%%%%%%%%%%%%%%%%%%%%%%%%%%%%%%%%%%%%%%%%
\section{Adaptive control algorithm}
Firstly,  we consider the design of the controller $u(k)$. To be specific, when the parameter $\theta$ is known, the optimal input $\bar{\phi}(k)\stackrel{\vartriangle}{=}[\bar{u}(k),\bar{u}(k-1)]^{\intercal}$ can be designed as:
\begin{equation}\label{multi-optimal-control-known}
	\bar{u}(k)=\frac{1}{\theta(1)}y^{*}(k)-\frac{\theta(2)}{\theta(1)}\bar{u}(k-1),\ k=1,2,...
\end{equation}
Secondly, when the parameter $\theta$ is unknown, it is first necessary to estimate the unknown parameter and then design the corresponding control input by the certainty equivalence principle. Therefore, for system  \eqref{multi-system}-\eqref{q-information} and index \eqref{multi-index}, inspired by the recursive projection identification algorithm in \cite{Guo-Zhao-1}, the adaptive control input can be designed as the following steps.

\textbf{Step 1.}(Initiation): Let $\hat{\theta}(0)\stackrel{\vartriangle}{=}[\hat{\theta}(1,0),\hat{\theta}(2,0)]^{\intercal}\in\Omega$, $\phi(1)=[u(1),0]^{\intercal}$, where $u(1)=\frac{y^{*}(1)}{\hat{\theta}(1,0)}$.

\textbf{Step 2.}(Weighted conversion of quantized observations): Based on $S(k)$ in \eqref{q-information}, let $\bar{S}(k)=\sum_{p=0}^m\beta_p {I}_{\{C_p<y(k)\leq C_{p+1}\}}$, where $\beta_{p}$ is a quantized weight with respect to set $(C_p,C_{p+1}]$ for $p=0,1,...,m$, which satisfies $\beta_0>\beta_1>\cdots>\beta_m$ and can be designed. 

\textbf{Step 3.}(Identification): The online stochastic approximation-type identification algorithm is as follows:
\begin{equation}\label{identification}
	\left\{
	\begin{aligned}
		\hat{\theta}(k)=&\Pi_{\Omega}\left(\hat{\theta}(k-1)+\frac{\phi(k)}{k}(A(k)-\bar{S}(k))\right),\\
		A(k)=&\Biggl(\sum_{p=0}^{m}\beta_{p}\biggl[{F}\left(C_{p+1}-\phi^{\intercal}(k)\hat{\theta}(k-1)\right)-{F}\left(C_{p}-\phi^{\intercal}(k)\hat{\theta}(k-1)\right)\biggl]\Biggl),
	\end{aligned}
	\right.
\end{equation}
where $\hat{\theta}(k)\stackrel{\vartriangle}{=}\left[\hat{\theta}(1,k),\hat{\theta}(2,k)\right]^{\intercal}$ is the estimation of $\theta$ at time $k$, $\Pi_{\Omega}(x)=\mbox{argmin}_{\kappa\in\Omega}\|x-\kappa\|,\ \forall x\in\mathbb{R}^{2}$ is the projection operator. 

\textbf{Step 4.}(Design of control input): By the certainty equivalent principle, the optimal input $\phi(k+1)=[u(k+1),u(k)]^{\intercal}$ in the case of unknown parameters is designed as follows:
\begin{equation}\label{multi-adaptive-control}
	u(k+1)=\frac{1}{\hat{\theta}(1,k)}y^{*}(k+1)-\frac{\hat{\theta}(2,k)}{\hat{\theta}(1,k)}u(k),\ k=1,2,...
\end{equation}

\textbf{Step 5.}(Update): Let $k=k+1$, the return to \textbf{Step 2}.
%%%%%%%%%%%%%%%%%%%%%%%%%%%%%%%%%%%%%%
\section{Main results}
In this section, the convergence of the estimation algorithm and the asymptotic validity of the adaptive tracking control are given. Firstly, we give some results about the projection operator and the input $\phi(k)$.

\begin{proposition}[Lemma 2.1 in \cite{Calamai-More} ]\label{pro-projection}
	The projection operator $\Pi_{\Omega}(\cdot)$ satisfies
	\[
	\|\Pi_{\Omega}(x_1)-\Pi_{\Omega}(x_2)\|\leq\|x_1-x_2\|,\quad \forall x_1,\ x_2\in\mathbb{R}^2.
	\]
\end{proposition}
%\begin{remark}
%According to Proposition \ref{pro-projection} and identification algorithm \eqref{identification}, for $\forall k\geq 1$, it results in that $\hat{\theta}(k)\in\Omega$.
%\end{remark}

\begin{proposition}\label{pro-1}
	For the input $\phi(k)$ designed by \eqref{multi-adaptive-control}, under Assumptions \ref{ass-reference-signal}-\ref{ass-multi-bounded}, we have the following two assertions.\\
	i)The uniformly bounded condition holds, i.e., there exists a positive constant $M$ such that
	\[
	\sup_{k\geq 1}\|\phi(k)\|\leq M.
	\]
	ii)The following persistent excitation condition holds.
	\[
	\lambda_{\min}\left(\sum_{i=k+2}^{k+h+1}{\phi}(i){\phi}^{\intercal}(i)\right)>\delta,\ \mbox{for}\ k\geq K_0,
	\]
	where $\delta\stackrel{\vartriangle}{=}\frac{\delta_y}{4(h-1)\bar{M}^2}$, $ \delta_y$, $h$ and $\bar{M}$ are defined in Assumptions \ref{ass-multi-reference-signal} and \ref{ass-multi-bounded}, respectively.
\end{proposition}

\begin{proof}
	%In the following proof we consider only the case $sgn(\theta(1))>0$. For the case $sgn(\theta(1))<0$ the proof can be given similarly.
	i) By \eqref{multi-adaptive-control}, the control input $u(k)$ can be expressed as
	\[
	u(k)=\frac{y^{*}(k)}{\hat{\theta}(1,k-1)}+\sum_{i=1}^{k-1}(-1)^{k-i}\frac{\prod_{j=i}^{k-1}\hat{\theta}(2,j)}{\prod_{j=i-1}^{k-1}\hat{\theta}(1,j)}y^{*}(i).
	\]
	Hence, for $\forall k\geq 1$, we have
	\begin{equation*}\label{u-bounded}
		|u(k)|\leq \frac{\bar{y}^{*}}{\underline{\theta}}+\sum_{i=1}^{k-1}\frac{\bar{\theta}^{k-i}}{\underline{\theta}^{k-i+1}}\bar{y}^{*}\leq \frac{\bar{y}^{*}}{\underline{\theta}-\bar{\theta}}.
	\end{equation*}
	Denote $M=\frac{\sqrt{2}\bar{y}^{*}}{\underline{\theta}-\bar{\theta}}$. It results in $\sup_{k\geq 1}\|\phi(k)\|\leq M$.
	
	ii) Denote $A_k(z)=1+\frac{\hat{\theta}(2,k-1)}{\hat{\theta}(1,k-1)}z$ and $B_k(z)=\frac{1}{\hat{\theta}(1,k-1)}$, where $z$ is the shift-back operator. Therefore, the control input \eqref{multi-adaptive-control} can be written as
	\begin{equation}\label{z-back-u}
		A_{k}(z)u(k)=B_k(z)y^{*}(k).
	\end{equation}
	Let $\psi(k)=A_k(z)\phi(k)=[B_k(z)y^{*}(k),B_{k-1}(z)y^{*}(k-1)+\Delta_{k-2}^{k-1}u(k-2)]^{\intercal}$, where
	\vspace{-5pt}
	\[
	\Delta_{k-2}^{k-1}\stackrel{\vartriangle}{=}\frac{\hat{\theta}(2,k-1)}{\hat{\theta}(1,k-1)}-\frac{\hat{\theta}(2,k-2)}{\hat{\theta}(1,k-2)}.
	\]
	\vspace{-5pt}
	On the one hand, for any $ x\in\mathbb{R}^2$, we obtain
	\begin{equation*}
		\begin{aligned}
			x^{\intercal}\left[\sum_{i=k+3}^{k+h+1}\psi(i)\psi^{\intercal}(i)\right]x
			=&\sum_{i=k+3}^{k+h+1}\left[x^{\intercal}\left(\phi(i)+\frac{\hat{\theta}(2,i-1)}{\hat{\theta}(1,i-1)}\phi(i-1)\right)\right]^2\\
			\leq&2\max\Bigl\{1,\frac{\bar{\theta}^2}{\underline{\theta}^2}\Bigl\}\left[\sum_{i=k+3}^{k+h+1}\sum_{j=1}^{2}\left[x^{\intercal}\phi(i-j+1)\right]^2\right]\\
			\leq& 2(h-1)\left[\sum_{i=k+2}^{k+h+1}\left[x^{\intercal}\phi(i)\phi^{\intercal}(i)x\right]\right].
		\end{aligned}	
	\end{equation*}
	\vspace{-5pt}
	Hence,
	\begin{equation}\label{PE-ineq-1}
		2(h-1)\lambda_{\min}\left(\sum_{i=k+2}^{k+h+1}\phi(i)\phi^{\intercal}(i)\right)\geq \lambda_{\min}\left(\sum_{i=k+3}^{k+h+1}\psi(i)\psi^{\intercal}(i)\right).
	\end{equation}
	On the other hand, we just need to prove the following inequality holds.
	\[
	\begin{aligned}
		\lambda_{\min}\left(\sum_{i=k+3}^{k+h+1}\psi(i)\psi^{\intercal}(i)\right)\geq K\lambda_{\min}\left(\sum_{i=k+3}^{k+h+1}\mathcal{Y}(i)\mathcal{Y}^{\intercal}(i)\right),\ \mbox{for\ some}\ K>0.
	\end{aligned}
	\]
	Define the following two backward-shifted polynomials with time-varying parameters
	\vspace{-5pt}
	\begin{equation}\label{H}
		\begin{aligned}
			H_{x,k}(z)=x(1)B_{k}(z)+x(2)B_{k-1}(z)z{=}\sum_{l=0}^{1}g_{l,k}(x)z^{l},\ \Delta_{x,k}(z)=x(2)\Delta_{k-2}^{k-1}z^2,
		\end{aligned}
	\end{equation}
	where $x=[x(1),x(2)]^{\intercal}$ and $g_k(x){=}[g_{0,k}(x),g_{1,k}(x)]^{\intercal}{=}\left[\frac{x(1)}{\hat{\theta}(1,k-1)},\frac{x(2)}{\hat{\theta}(1,k-2)}\right]^{\intercal}.$
	
	Firstly, by \eqref{H}, we have $x^{\intercal}\psi(i)=H_{x,i}(z)y^{*}(i)+\Delta_{x,i}(z)u(i)$. Then, we have
	\begin{equation}\label{ineq-1}
		\begin{aligned}
			\sum_{i=k+3}^{k+h+1}(x^{\intercal}\psi(i))^2
			=&\sum_{i=k+3}^{k+h+1}\biggl[H_{x,i}(z)y^{*}(i)+x(2)\Delta_{i-2}^{i-1}u(i-2)\biggl]^2\\
			\geq&2\Biggl[\sum_{i=k+3}^{k+h+1}H_{x,i}(z)y^{*}(i)x(2)\Delta_{i-2}^{i-1}u(i-2)  \Biggl]+\sum_{i=k+3}^{k+h+1}\biggl[\sum_{p=0}^{1}g_{p,i}(x)y^{*}(i-p)\biggl]^2\\
			=&\sum_{i=k+3}^{k+h+1}\biggl[g^{\intercal}_{i}(x)\mathcal{Y}(i)\biggl]^2+2\Biggl[\sum_{i=k+3}^{k+h+1}H_{x,i}(z)y^{*}(i)x(2)\Delta_{i-2}^{i-1}u(i-2)  \Biggl]\\
			\geq&2\Biggl[\sum_{i=k+3}^{k+h+1}g^{\intercal}_{k+2}(x)\mathcal{Y}(i)\mathcal{Y}^{\intercal}(i)(g_{i}(x)-g_{k+2}(x)) \Biggl]+2\Biggl[\sum_{i=k+3}^{k+h+1}H_{x,i}(z)y^{*}(i)x(2)\Delta_{i-2}^{i-1}u(i-2)  \Biggl]\\
			&+g^{\intercal}_{k+2}(x)\sum_{i=k+3}^{k+h+1}\biggl[\mathcal{Y}(i)\mathcal{Y}^{\intercal}(i)\biggl]g_{k+2}(x).
		\end{aligned}
	\end{equation}
We just need to deal with the last two terms of \eqref{ineq-1} separately. Since
	\[
	\|g_{k+2}(x)\|\leq\sqrt{\frac{(x(1))^2}{\underline{\theta}^2}+\frac{(x(2))^2}{\underline{\theta}^2}}=\frac{\|x\|}{\underline{\theta}}, \|\mathcal{Y}(i)\|\leq\sqrt{2}\bar{y}^{*},
	\]
	then, denote $K_1=\frac{2(\bar{y}^{*})^2\|x\|}{\underline{\theta}}$, the following relation holds.
	\begin{equation}\label{ineq-1.1}
		\begin{aligned}
			&\sum_{i=k+3}^{k+h+1}\Bigl|g^{\intercal}_{k+2}(x)\mathcal{Y}(i)\mathcal{Y}^{\intercal}(i)(g_{i}(x)-g_{k+2}(x))\Bigl|\\	
			\leq& \Biggl[\sum_{i=k+3}^{k+h+1}\|g^{\intercal}_{k+2}(x)\|\times\|\mathcal{Y}(i)\|\times\|\mathcal{Y}^{\intercal}(i)\|\times\|(g_{i}(x)-g_{k+2}(x))\| \Biggl]\\
			\leq&K_1	\Biggl[\sum_{i=k+3}^{k+h+1}\sqrt{\Biggl(\frac{x(1)}{\hat{\theta}(1,i-1)}-\frac{x(1)}{\hat{\theta}(1,k+1)}\Biggl)^2+\Biggl(\frac{x(2)}{\hat{\theta}(1,i-2)}-\frac{x(2)}{\hat{\theta}(1,k)}\Biggl)^2} \Biggl]\\
			=&K_1	\Biggl[\sum_{i=k+3}^{k+h+1}\sqrt{\Biggl(\frac{x(1)(\hat{\theta}(1,k+1)-\hat{\theta}(1,i-1))}{\hat{\theta}(1,i-1)\hat{\theta}(1,k+1)}\Biggl)^2+\Biggl(\frac{x(2)(\hat{\theta}(1,k)-\hat{\theta}(1,i-2))}{\hat{\theta}(1,i-2)\hat{\theta}(1,k)}\Biggl)^2} \Biggl]\\
			%\leq& \frac{K_1\sqrt{\max_{i=1,2}\{x^2(i)\}}}{\underline{\theta}^2}\Biggl[\sum_{i=k+3}^{k+h+1}\sqrt{\sum_{j=1}^2\bigl(\hat{\theta}(1,k+2-j)-\hat{\theta}(1,i-j)\bigl)^2}\Biggl]\\
			\leq&\frac{K_1\sqrt{\max_{i=1,2}\{x^2(i)\}}}{\underline{\theta}^2}\Biggl[\sum_{i=k+3}^{k+h+1}\sqrt{\sum_{j=1}^2\bigl\|\hat{\theta}(k+2-j)-\hat{\theta}(i-j)\bigl\|^2}\Biggl].
		\end{aligned}
	\end{equation}
	Denote $K_2=\frac{2(\bar{y}^{*})^2\|x\|\sqrt{\max_{i=1,2}\{x^2(i)\}}}{\underline{\theta}^3}$. Since $||\hat{\theta}(i)-\hat{\theta}(j)||^2\leq
	\frac{M^2(i-j)^2}{(j+1)^2}\bigl(m\sum_{p=0}^{m}\beta^2_p\bigl), j\leq i$, then we have
	\begin{equation}\label{ineq-1.2}
		\begin{aligned}
			\sum_{i=k+3}^{k+h+1}\Bigl|g^{\intercal}_{k+2}(x)\mathcal{Y}(i)\mathcal{Y}^{\intercal}(i)(g_{i}(x)-g_{k+2}(x))\Bigl|\leq& K_2\Biggl[\sum_{i=k+3}^{k+h+1}\sqrt{\sum_{j=1}^2\frac{M^2(i-k-2)^2}{(k+3-j)^2}m\sum_{p=0}^{m}\beta^2_p}\Biggl]\\\leq&\frac{K_2M(h-1)^{\frac{3}{2}}\sqrt{2m\sum_{p=0}^{m}\beta^2_p}}{k+2}	=O\Bigl(\frac{1}{k}\Bigl).
		\end{aligned}
	\end{equation}
	Next, we estimate the last term of \eqref{ineq-1}. 
	\begin{equation}\label{ineq-1.3}
		\begin{aligned}
			&	\sum_{i=k+3}^{k+h+1}\biggl| H_{x,i}(z)y^{*}(i)x(2)\Delta_{i-2}^{i-1}u(i-2)\biggl|\\
			\leq&\frac{M\bar{y}^{*}(|x(1)+|x(2)|)|x(2)|}{\underline{\theta}}	\Biggl[\sum_{i=k+3}^{k+h+1}|\Delta_{i-2}^{i-1}| \Biggl]\\
			=&\frac{M\bar{y}^{*}(|x(1)+|x(2)|)|x(2)|}{\underline{\theta}}\sum_{i=k+3}^{k+h+1}\Biggl[\Biggl|\frac{\hat{\theta}(2,i-1)}{\hat{\theta}(1,i-1)}-\frac{\hat{\theta}(2,i-2)}{\hat{\theta}(1,i-2)} \Biggl|\Biggl]\\
			=&\frac{M\bar{y}^{*}(|x(1)+|x(2)|)|x(2)|}{\underline{\theta}}\sum_{i=k+3}^{k+h+1}\Biggl[\frac{|\hat{\theta}(1,i-2)\hat{\theta}(2,i-1)-\hat{\theta}(1,i-1)\hat{\theta}(2,i-2)|}{|\hat{\theta}(1,i-1)\hat{\theta}(1,i-2)|} \Biggl]\\
			\leq&\frac{M\bar{y}^{*}(|x(1)+|x(2)|)|x(2)|}{\underline{\theta}^3}\sum_{i=k+3}^{k+h+1}\Bigl[|\hat{\theta}(1,i-2)||\hat{\theta}(2,i-1)-\hat{\theta}(2,i-2)|+|\hat{\theta}(2,i-2)||\hat{\theta}(1,i-1)-\hat{\theta}(1,i-2)| \Bigl]\\
			\leq&\frac{\sqrt{2}M\bar{M}\bar{y}^{*}(|x(1)+|x(2)|)|x(2)|}{\underline{\theta}^3}\sum_{i=k+3}^{k+h+1}\Bigl[\|\hat{\theta}(i-1)-\hat{\theta}(i-2)\|\Bigl].
		\end{aligned}
	\end{equation}
	Since $||\hat{\theta}(i)-\hat{\theta}(j)||\leq
	\frac{M(i-j)}{(j+1)}\sum_{p=0}^{m}|\beta_p|,\ j\leq i$, then
	\begin{equation}\label{ineq-1.4}
		\begin{aligned}
			\sum_{i=k+3}^{k+h+1}\biggl| H_{x,i}(z)y^{*}(i)x(2)\Delta_{i-2}^{i-1}u(i-2)\biggl|\leq& \frac{\sqrt{2}M\bar{M}\bar{y}^{*}(|x(1)+|x(2)|)|x(2)|}{\underline{\theta}^3}\sum_{i=k+3}^{k+h+1}\frac{M}{i-1}\sum_{p=0}^{m}|\beta_p|\\
			\leq&\frac{\sqrt{2}(h-1)M^2\bar{M}\bar{y}^{*}(|x(1)+|x(2)|)|x(2)|}{\underline{\theta}^3(k+2)}\sum_{p=0}^{m}|\beta_p|\\
			=&O\Bigl(\frac{1}{k}\Bigl).
		\end{aligned}
	\end{equation}
	Finally, let $x$ be the unit vector. Then $\|g_k(x)\|^2=\sum_{i=1}^2\frac{x^2(i)}{\hat{\theta}^2(1,k-i)}\in\left[\frac{1}{\bar{M}^2}, \frac{1}{\underline{\theta}^2}\right]$. Taking \eqref{ineq-1.1}-\eqref{ineq-1.4} into \eqref{ineq-1}, we obtain the following relations.
	\begin{equation}\label{PE-ineq}
		\begin{aligned}
			x^{\intercal}\left[\sum_{i=k+3}^{k+h+1}\psi(i)\psi^{\intercal}(i)\right]x\geq&g^{\intercal}_{k+2}(x)\sum_{i=k+3}^{k+h+1}\biggl[\mathcal{Y}(i)\mathcal{Y}^{\intercal}(i)\biggl]g_{k+2}(x)-O\Bigl(\frac{1}{k}\Bigl)\\
			\geq&\frac{\lambda_{\min}\left(\sum_{i=k+3}^{k+h+1}\mathcal{Y}(i)\mathcal{Y}^{\intercal}(i)\right)}{\bar{M}^2}-O\left(\frac{1}{k}\right).
		\end{aligned}	
	\end{equation}
	Therefore, there exists $K_0>0$ such that for $ k\geq K_0$,
	\begin{equation}\label{PE-ineq-1.1}
		x^{\intercal}\left[\sum_{i=k+3}^{k+h+1}\psi(i)\psi^{\intercal}(i)\right]x\geq \frac{\lambda_{\min}\left(\sum_{i=k+3}^{k+h+1}\mathcal{Y}(i)\mathcal{Y}^{\intercal}(i)\right)}{2\bar{M}^2}.
	\end{equation}
	Therefore, based on \eqref{PE-ineq-1.1}, we obtain
	\begin{equation}\label{PE-ineq-2} \lambda_{\min}\left(\sum_{i=k+3}^{k+h+1}\psi(i)\psi^{\intercal}(i)\right)\geq\frac{\delta_y}{2\bar{M}^2}, \ \mbox{for}\ k\geq K_0.
	\end{equation}
	Therefore, according to \eqref{PE-ineq-1} and \eqref{PE-ineq-2}, we obtain
	\[
	\lambda_{\min}\left(\sum_{i=k+2}^{k+h+1}\phi(i)\phi^{\intercal}(i)\right)\geq\frac{\delta_y}{4(h-1)\bar{M}^2}, \ \mbox{for}\ k\geq K_0.
	\]
	This completes the proof.
\end{proof}

\begin{lemma}[\cite{Zhao-Wang-Bi-1}]\label{lem-convergence}
	For any given $\zeta\in\mathbb{R}$, $\gamma\in\mathbb{R}$ and $r\in \mathbb{Z}^{+}$, we have the following assertions as $k\rightarrow\infty$:
	\begin{equation}\label{lem-convergence-1}
		\prod_{i=r}^{k}\left(1-\frac{\zeta}{i}\right)=O\left(\frac{1}{k^{\zeta}}\right),\nonumber
	\end{equation}
	and
	\begin{equation}\label{lem-convergence-2}
		\sum_{l=r}^{k-1}\prod_{i=l+1}^{k}\left(1-\frac{\zeta}{i}\right)\frac{1}{l^{\gamma}}=
		\left\{
		\begin{aligned}
			&O\left(\frac{1}{k^{\gamma-1}}\right), &\gamma<\zeta+1,\\
			&O\left(\frac{\ln (k)}{k^{\zeta}}\right),  &\gamma=\zeta+1,\\
			&O\left(\frac{1}{k^{\zeta}}\right), &\gamma>\zeta+1.
		\end{aligned}
		\right.\nonumber
	\end{equation}
\end{lemma}

In addition, for system noise $\{w(k)\}_{k\geq 1}$, we need the following additional assumption in the framework of this section.
\begin{assumption}\label{ass-multi-noise}
	The probability density function of $w(1)$ denoted by $f(\cdot)$. Besides, there exists a minimum $f^{*}$ of $f(\cdot)$ on interval $[-D_1,D_1]$, where $D_1=\max\{|C_m|,|C_1|\}+M\bar{M}$.	
\end{assumption}

Denote the estimation error $\tilde{\theta}(k){=}\hat{\theta}(k)-\theta$. Then, the following results give the convergence property of the estimation error $\tilde{\theta}(k)$.
\begin{theorem}\label{theorem-convergence-2}
	Under Assumptions \ref{ass-reference-signal}-\ref{ass-multi-noise}, the mean square convergence rate of the estimation error is given as follows:
	\begin{equation*}
		\mathbb{E}\left[\left\|\tilde{\theta}(k)\right\|^2\right]=
		\left\{
		\begin{aligned}
			&O\left(\frac{1}{k}\right),\ &\mbox{if}\ \frac{2(\beta_0-\beta_m)\delta {f}^{*}}{{h}}>1;\\
			&O\left(\frac{\ln (k)}{k}\right),\ &\mbox{if}\ \frac{2(\beta_0-\beta_m)\delta {f}^{*}}{{h}}=1;\\
			&O\left(\frac{1}{k^{\frac{2(\beta_0-\beta_m)\delta {f}^{*}}{{h}}}}\right),\ &\mbox{if}\ \frac{2(\beta_0-\beta_m)\delta {f}^{*}}{{h}}<1,
		\end{aligned}
		\right.
	\end{equation*}
	where $h$, $f^{*}$ and $\delta$ are defined in Assumptions \ref{ass-multi-reference-signal}, \ref{ass-multi-noise} and Proposition \ref{pro-1}, respectively.
\end{theorem}
\begin{proof}
	By \eqref{identification} and Proposition \ref{pro-projection}, we have
	\begin{equation}\label{ineq-4.1}
		\begin{aligned}
			\left\|\tilde{\theta}(k)\right\|^2=&\left\|\Pi_{\Omega}\left(\hat{\theta}(k-1)+\frac{\phi(k)}{k}(A(k)-\bar{S}(k))\right)-\Pi_{\Omega}(\theta)\right\|^2\\
			\leq&\left\|\tilde{\theta}(k-1)+\frac{\phi(k)}{k}(A(k)-\bar{S}(k))\right\|^2\\
			\leq&\biggl\|\tilde{\theta}(k-1)+\frac{\phi(k)}{k}\biggl(\sum_{p=0}^m\beta_p\biggl[{F}\left(C_{p+1}-\phi^{\intercal}(k)\hat{\theta}(k-1)\right)-{F}\left(C_{p}-\phi^{\intercal}(k)\hat{\theta}(k-1)\right)\biggl]-\bar{S}(k)\biggl)\biggl\|^2\\
			=&\left\|\tilde{\theta}(k-1)\right\|^2+\frac{\|\phi(k)\|^2}{k^2}\biggl(\sum_{p=0}^m\beta_p\biggl[{F}\left(C_{p+1}-\phi^{\intercal}(k)\hat{\theta}(k-1)\right)-{F}\left(C_{p}-\phi^{\intercal}(k)\hat{\theta}(k-1)\right)\biggl]-\bar{S}(k)\biggl)^2\\
			&+\frac{2\tilde{\theta}^{\intercal}(k-1)\phi(k)}{k}\Biggl(\sum_{p=0}^m\beta_p\biggl[{F}\left(C_{p+1}-\phi^{\intercal}(k)\hat{\theta}(
			k-1)\right)-{F}\left(C_{p}-\phi^{\intercal}(k)\hat{\theta}(k-1)\right)\biggl]-\bar{S}(k)\Biggl).
		\end{aligned}	
	\end{equation}
	
	Denote $\mathcal{F}_{k-1}=\sigma\{w(s),\ s\leq k-1\}$. By \eqref{identification}-\eqref{multi-adaptive-control}, we obtain $\tilde{\theta}(k-1)$ and $\phi(k)$ are $\mathcal{F}_{k-1}$-measurable. Taking conditional expectation with respect to $\bar{S}(k)$, we obtain
	\begin{equation}\label{ineq-4.4}
		\begin{aligned}
			\mathbb{E}\left[\bar{S}(k)|\mathcal{F}_{k-1}\right]=&\mathbb{E}\left[\sum_{p=0}^m\beta_pI_{\{C_p<y(k)\leq C_{p+1}\}}\biggl|\mathcal{F}_{k-1}\right]\\
			=&\mathbb{E}\left[\sum_{p=0}^m\beta_pI_{\{C_p-\phi^{\intercal}(k)\theta<w(k)\leq C_{p+1}-\phi^{\intercal}(k)\theta\}}\biggl|\mathcal{F}_{k-1}\right]\\
			=&\sum_{p=0}^m\beta_p\left[{F}\left(C_{p+1}-\phi^{\intercal}(k){\theta}\right)-{F}\left(C_{p}-\phi^{\intercal}(k){\theta}\right)\right].
		\end{aligned}
	\end{equation}	
	Therefore, 	based on \eqref{ineq-4.4}, we have
	\begin{equation}\label{ineq-4.5}
		\begin{aligned}
			&\mathbb{E}\Biggl[\frac{2\tilde{\theta}^{\intercal}(k-1)\phi(k)}{k}\biggl(\sum_{p=0}^m\beta_p\biggl[{F}\left(C_{p+1}-\phi^{\intercal}(k)\hat{\theta}(k-1)\right)-{F}\left(C_{p}-\phi^{\intercal}(k)\hat{\theta}(k-1)\right)\biggl]-\bar{S}(k)\biggl)\biggl|\mathcal{F}_{k-1}\Biggl]\\
			%	=&\frac{2\tilde{\theta}^{\intercal}(k-1)\phi(k)}{k}\Biggl[\sum_{p=0}^m\beta_p\biggl[{F}\left(C_{p+1}-\phi^{\intercal}(k)\hat{\theta}(k-1)\right)\\
			%	&-{F}\left(C_{p+1}-\phi^{\intercal}(k){\theta}\right)\biggl]-\sum_{p=0}^m\beta_p\biggl[{F}\left(C_{p}-\phi^{\intercal}(k)\hat{\theta}(k-1)\right)\\
			%	&-{F}\left(C_{p}-\phi^{\intercal}(k){\theta}\right)\biggl]\Biggl]\\
			=&\frac{2\tilde{\theta}^{\intercal}(k-1)\phi(k)}{k}\Biggl[\sum_{p=0}^{m-1}\beta_p\biggl[{F}\left(C_{p+1}-\phi^{\intercal}(k)\hat{\theta}(k-1)\right)-{F}\left(C_{p+1}-\phi^{\intercal}(k){\theta}\right)\biggl]-\sum_{p=1}^m\beta_p\biggl[{F}\left(C_{p}-\phi^{\intercal}(k)\hat{\theta}(k-1)\right)\\
			&-{F}\left(C_{p}-\phi^{\intercal}(k){\theta}\right)\biggl]\Biggl]\\
			=&\frac{2\tilde{\theta}^{\intercal}(k-1)\phi(k)}{k}\Biggl[\sum_{p=1}^{m}\beta_{p-1}\biggl[{F}\left(C_{p}-\phi^{\intercal}(k)\hat{\theta}(k-1)\right)-{F}\left(C_{p}-\phi^{\intercal}(k){\theta}\right)\biggl]-\sum_{p=1}^m\beta_p\biggl[{F}\left(C_{p}-\phi^{\intercal}(k)\hat{\theta}(k-1)\right)\\
			&-{F}\left(C_{p}-\phi^{\intercal}(k){\theta}\right)\biggl]\Biggl]\\
			=&\frac{2\tilde{\theta}^{\intercal}(k-1)\phi(k)}{k}\sum_{p=1}^m(\beta_{p-1}-\beta_p)\biggl[{F}\left(C_{p}-\phi^{\intercal}(k)\hat{\theta}(k-1)\right)-{F}\left(C_{p}-\phi^{\intercal}(k){\theta}\right)\biggl].
		\end{aligned}
	\end{equation}
	According to the mean value theorem, there exists $\eta_p(k)\in(C_{p}-\phi^{\intercal}(k)\hat{\theta}(k-1),C_{p}-\phi^{\intercal}(k){\theta})$ or $(C_{p}-\phi^{\intercal}(k){\theta},C_{p}-\phi^{\intercal}(k)\hat{\theta}(k-1))$ such that
	\begin{equation}\label{mean-value-4}
		\begin{aligned}
			{F}\left(C_{p}-\phi^{\intercal}(k)\hat{\theta}(k-1)\right)-{F}(C_{p}-\phi^{\intercal}(k){\theta})=-{f(\eta_p(k))}\phi^{\intercal}(k)\tilde{\theta}(k-1).
		\end{aligned}
	\end{equation}
	%\vspace{-20pt}
	Thus, based on \eqref{mean-value-4} and Assumption \ref{ass-multi-noise}, \eqref{ineq-4.5} can be rewritten as
	\begin{equation}\label{ineq-4.6}
		\begin{aligned}
			&\mathbb{E}\Biggl[\frac{2\tilde{\theta}^{\intercal}(k-1)\phi(k)}{k}\biggl(\sum_{p=0}^m\beta_p\biggl[{F}\left(C_{p+1}-\phi^{\intercal}(k)\hat{\theta}(k-1)\right)-{F}\left(C_{p}-\phi^{\intercal}(k)\hat{\theta}(k-1)\right)\biggl]-\bar{S}(k)\biggl)\biggl|\mathcal{F}_{k-1}\Biggl]\\
			=&\frac{2\tilde{\theta}^{\intercal}(k-1)\phi(k)}{k}\sum_{p=1}^m(\beta_{p-1}-\beta_p)[-{f(\eta_p(k))}\phi^{\intercal}(k)\tilde{\theta}(k-1)]\\
			\leq&-\frac{2(\beta_{0}-\beta_m)f^{*}\tilde{\theta}^{\intercal}(k-1)\phi(k)\phi^{\intercal}(k)\tilde{\theta}(k-1)}{k}.
		\end{aligned}
	\end{equation}
	By \eqref{ineq-4.6}, \eqref{ineq-4.1} can be written as
	\begin{equation}\label{k-1-ineq}
		\begin{aligned}
			\mathbb{E}\left[\left\|\tilde{\theta}(k)\right\|^2\right]\leq \mathbb{E}\left[\left\|\tilde{\theta}(k-1)\right\|^2\right]+\frac{M^2\left(m\sum_{p=0}^m\beta^2_{p}\right)}{k^2}-\frac{2(\beta_0-\beta_m) f^{*}}{k}\mathbb{E}\left[\tilde{\theta}^{\intercal}(k-1)\phi(k)\phi^{\intercal}(k)\tilde{\theta}(k-1)\right].	
		\end{aligned}	
	\end{equation}
	%%%%%%%%%%%%%%%%%%%%%%%%%%%%%%%%%%%%%%%%%%%%%%%%%%%%%%%%%%%%%%%%%%%%%%%%%%%%
	Consequently, we expand $\left\|\tilde{\theta}(k)\right\|^2$ to $\Bigl\|\tilde{\theta}(k-{h})\Bigl\|^2$ by \eqref{k-1-ineq} and obtain
	\begin{equation}\label{k-h-ineq}
		\begin{aligned}
			\mathbb{E}\left[\left\|\tilde{\theta}(k)\right\|^2\right]\leq \mathbb{E}\left[\left\|\tilde{\theta}(k-{h})\right\|^2\right]+\sum_{i=k-{h}+1}^k\frac{M^2\left(m\sum_{p=0}^m\beta^2_{p}\right)}{i^2}-\sum_{i=k-{h}+1}^k\frac{2(\beta_0-\beta_m) f^{*}}{i}\mathbb{E}\left[\tilde{\theta}^{\intercal}(i-1)\phi(i)\phi^{\intercal}(i)\tilde{\theta}(i-1)\right].
		\end{aligned}
	\end{equation}
	Since for $i=k-h+1,...,k$, we have
	\begin{equation*}
		\begin{aligned}
			&\tilde{\theta}^{\intercal}(i-1)\phi(i)\phi^{\intercal}(i)\tilde{\theta}(i-1)\\
			=&\left(\tilde{\theta}(i-1)-\tilde{\theta}(k-{h})+\tilde{\theta}(k-{h})\right)^{\intercal}\phi(i)\phi^{\intercal}(i)\left(\tilde{\theta}(i-1)-\tilde{\theta}(k-{h})+\tilde{\theta}(k-{h})\right)\\
			=&\tilde{\theta}(k-{h})^{\intercal}\phi(i)\phi^{\intercal}(i)\tilde{\theta}(k-{h})+\left(\tilde{\theta}(i-1)-\tilde{\theta}(k-{h})\right)^{\intercal}\phi(i)\phi^{\intercal}(i)\left(\tilde{\theta}(i-1)-\tilde{\theta}(k-{h})\right)\\
			&+2\left(\tilde{\theta}(i-1)-\tilde{\theta}(k-{h})\right)^{\intercal}\phi(i)\phi^{\intercal}(i)\tilde{\theta}(k-{h}),
		\end{aligned}
	\end{equation*}
	and Proposition \ref{pro-1}, then we have
	\begin{equation}\label{k-h-ineq-1.1}
		\begin{aligned}
			&-\sum_{i=k-{h}+1}^{k}\frac{1}{i}\mathbb{E}\left[\tilde{\theta}^{\intercal}(i-1)\phi(i)\phi^{\intercal}(i)\tilde{\theta}(i-1)\right]\\
			=&-\sum_{i=k-{h}+1}^{k}\frac{1}{i}\mathbb{E}\left[\tilde{\theta}(k-{h})^{\intercal}\phi(i)\phi^{\intercal}(i)\tilde{\theta}(k-{h})\right]-\sum_{i=k-{h}+1}^{k}\frac{2}{i}\mathbb{E}\left[\left(\tilde{\theta}(i-1)-\tilde{\theta}(k-{h})\right)^{\intercal}\phi(i)\phi^{\intercal}(i)\tilde{\theta}(k-{h})\right]\\
			&-\sum_{i=k-{h}+1}^{k}\frac{1}{i}\mathbb{E}\Bigl[\left(\tilde{\theta}(i-1)-\tilde{\theta}(k-{h})\right)^{\intercal}\phi(i)\phi^{\intercal}(i)\left(\tilde{\theta}(i-1)-\tilde{\theta}(k-{h})\right)\Bigl]\\
			\leq&-\sum_{i=k-{h}+1}^{k}\frac{1}{i}\mathbb{E}\left[\tilde{\theta}(k-{h})^{\intercal}\phi(i)\phi^{\intercal}(i)\tilde{\theta}(k-{h})\right]-\sum_{i=k-{h}+1}^{k}\frac{2}{i}\mathbb{E}\left[\left(\tilde{\theta}(i-1)-\tilde{\theta}(k-{h})\right)^{\intercal}\phi(i)\phi^{\intercal}(i)\tilde{\theta}(k-{h})\right]\\
			\leq& \frac{2\bar{M}M^2}{k-{h}+1}\sum_{i=k-{h}+1}^{k}\sqrt{\mathbb{E}\left[\left\|\tilde{\theta}(i-1)-\tilde{\theta}(k-{h})\right\|^2\right]}-\frac{\delta}{k}\mathbb{E}\left[\left\|\tilde{\theta}(k-{h})\right\|^2\right]\\
			\leq&-\frac{\delta}{k}\mathbb{E}\left[\left\|\tilde{\theta}(k-{h})\right\|^2\right]+\frac{2\bar{M}({h}-1){h} M^3}{(k-{h})^2}\sqrt{m\sum_{p=0}^{m}\beta^2_{p}}.	
		\end{aligned}
	\end{equation}
	Therefore, by \eqref{k-h-ineq-1.1}, \eqref{k-h-ineq} can be expressed as
	\begin{equation}\label{ex-re-inequality-3}
		\begin{aligned}
			\mathbb{E}\left[\left\|\tilde{\theta}(k)\right\|^2\right]\leq& \left(1-\frac{2\delta(\beta_0-\beta_m) f^{*}}{k}\right)\mathbb{E}\left[\left\|\tilde{\theta}(k-{h})\right\|^2\right]+\frac{{h}M^2}{(k-{h}+1)^2}{\left(m\sum_{p=0}^m\beta^2_{p}\right)}\\
			&+{\frac{4\bar{M}({h}-1){h}M^3(\beta_0-\beta_m)f^{*} M^3}{(k-{h})^2}\sqrt{m\sum_{p=0}^{m}\beta^2_{p}}}\\
			\leq& \left(1-\frac{2\delta(\beta_0-\beta_m) f^{*}}{k}\right)\mathbb{E}\left[\left\|\tilde{\theta}(k-{h})\right\|^2\right]+O\Biggl(\frac{1}{(k-{h})^{2}}\Biggl).
		\end{aligned}
	\end{equation}
	According to the inequality \eqref{ex-re-inequality-3}, expanding $\left\|\tilde{\theta}(k)\right\|^2$ to $\Bigl\|\tilde{\theta}(k-\lfloor \frac{k}{{h}}\rfloor {h})\Bigl\|^2$ yields
	
	\begin{equation}\label{ex-re-inequality-4}
		\begin{aligned}
			\mathbb{E}\left[\left\|\tilde{\theta}(k)\right\|^2\right]\leq& \prod^{\lfloor \frac{k}{{h}}\rfloor-1}_{j=0}\left(1-\frac{2\delta(\beta_0-\beta_m) f^{*}}{k-j{h}}\right)\mathbb{E}\left[\left\|\tilde{\theta}\left(k-\left\lfloor \frac{k}{{h}}\right\rfloor {h}\right)\right\|^2\right]\\
			&+\sum_{j=2}^{\lfloor \frac{k}{{h}}\rfloor}\prod_{i=0}^{j-2}\left(1-\frac{2\delta(\beta_0-\beta_m) f^{*}}{k-i{h}}\right)O\left(\frac{1}{(k-j{h}+1)^{2}}\right)+O\left(\frac{1}{(k-{h})^{2}}\right).
		\end{aligned}
	\end{equation}
	Firstly, the first item of the right side of \eqref{ex-re-inequality-4} is
	\[
	\begin{aligned}
		\prod^{\lfloor \frac{k}{{h}}\rfloor-1}_{j=0}\left(1-\frac{2\delta(\beta_0-\beta_m) f^{*}}{k-j{h}}\right)\leq \prod^{\lceil \frac{k}{{h}}\rceil}_{m=\lceil \frac{k}{{h}}\rceil-\lfloor \frac{k}{{h}}\rfloor+1}\left(1-\frac{2\delta(\beta_0-\beta_m) f^{*}}{m{h}}\right).
	\end{aligned}
	\]
Secondly, the second item of the right side of \eqref{ex-re-inequality-4} is
	\[	
	\begin{aligned}
		&\sum_{j=2}^{\lfloor\frac{k}{{h}}\rfloor}\prod_{i=0}^{j-2}\left(1-\frac{2\delta(\beta_0-\beta_m) f^{*}}{k-i{h}}\right)O\left(\frac{1}{(k-j{h}+1)^{2}}\right)\nonumber\\
		\leq& \sum_{j=2}^{\lfloor \frac{k}{{h}}\rfloor}\prod_{i=0}^{j-2}\left(1-\frac{2\delta(\beta_0-\beta_m) f^{*}}{k-i{h}}\right)O\left(\frac{1}{(\lfloor \frac{k}{{h}}\rfloor-j+\frac{1}{{h}})^{2}h^{2}}\right)\nonumber\\
		\leq& \sum_{m=1}^{\lfloor \frac{k}{{h}}\rfloor-2}\prod^{\lceil \frac{k}{{h}}\rceil}_{p=\lceil \frac{k}{{h}}\rceil-\lfloor \frac{k}{{h}}\rfloor+m+2}\left(1-\frac{2\delta(\beta_0-\beta_m) f^{*}}{p{h}}\right)O\left(\frac{1}{m^{2}{h}^{2}}\right)\nonumber+O\left(\prod_{i=0}^{\lfloor \frac{k}{{h}}\rfloor-2}\left(1-\frac{2\delta(\beta_0-\beta_m)f^{*}}{k-i{h}}\right)\right).\nonumber
	\end{aligned}
	\]
	Thus, we have
		\begin{equation}\label{ex-re-inequality-5}
			\begin{aligned}
				\mathbb{E}\left[\left\|\tilde{\theta}(k)\right\|^2\right]\leq&\prod^{\lceil \frac{k}{{h}}\rceil}_{m=\lceil \frac{k}{{h}}\rceil-\lfloor \frac{k}{{h}}\rfloor+1}\left(1-\frac{2\delta(\beta_0-\beta_m) f^{*}}{m{h}}\right)\mathbb{E}\left[\left\|\tilde{\theta}(k-\left\lfloor \frac{k}{{h}}\right\rfloor {h})\right\|^2\right]\\
				&+\sum_{m=1}^{\lfloor \frac{k}{{h}}\rfloor-2}\prod^{\lceil \frac{k}{{h}}\rceil}_{p=\lceil \frac{k}{{h}}\rceil-\lfloor \frac{k}{{h}}\rfloor+m+2}\left(1-\frac{2\delta(\beta_0-\beta_m) f^{*}}{p{h}}\right)O\left(\frac{1}{m^{2}{h}^{2}}\right)\\
				&+O\left(\prod_{i=0}^{\lfloor \frac{k}{{h}}\rfloor-2}\left(1-\frac{2\delta(\beta_0-\beta_m) f^{*}}{k-i{h}}\right)\right)+O\left(\frac{1}{(k-{h})^{2}}\right).\nonumber
			\end{aligned}
		\end{equation}
	%%%%%%%%%%%%%%%%%%%%%%%%%%%%%%%%%%%%%%%%%%%%%%%%%%%%%%%%%%%%%%%%%%%%%%%%%%%%
	Therefore, by Lemma \ref{lem-convergence}, we finally obtain
		\begin{equation*}
			\mathbb{E}\left[\left\|\tilde{\theta}(k)\right\|^2\right]=
			\left\{
			\begin{aligned}
				&O\left(\frac{1}{k}\right),\ &\mbox{if}\ \frac{2(\beta_0-\beta_m)\delta {f}^{*}}{{h}}>1;\\
				&O\left(\frac{\ln (k)}{k}\right),\ &\mbox{if}\ \frac{2(\beta_0-\beta_m)\delta {f}^{*}}{{h}}=1;\\
				&O\left(\frac{1}{k^{\frac{2(\beta_0-\beta_m)\delta {f}^{*}}{{h}}}}\right),\ &\mbox{if}\ \frac{2(\beta_0-\beta_m)\delta {f}^{*}}{{h}}<1,
			\end{aligned}
			\right.
		\end{equation*}
	This completes the proof.
\end{proof}

\begin{theorem}\label{multi-convergence}
	Under Assumptions \ref{ass-reference-signal}-\ref{ass-multi-noise}, the estimation in \eqref{identification} will convergence to the real parameter $\theta$ in the sense of mean square and almost surely
	\[
	\lim_{k\rightarrow\infty}\mathbb{E}\left[\left\|\tilde{\theta}(k)\right\|^2\right]=0,\quad \lim_{k\rightarrow\infty}\hat{\theta}(k)=\theta,\ a.s.
	\]
\end{theorem}
\begin{proof}
	By Theorem \ref{theorem-convergence-2}, it reduces that 
	\begin{equation}\label{eq-4.1}
		\lim_{k\rightarrow\infty}\mathbb{E}\left[\left\|\tilde{\theta}(k)\right\|^2\right]=0.
	\end{equation}
	Since $\phi(k)$ and $\hat{\theta}(k-1)$ are $\mathcal{F}_{k-1}$-measurable, we have by \eqref{k-1-ineq}
	\begin{equation}\label{ineq-4.8}
		\begin{aligned}
			\mathbb{E}\left[\left\|\tilde{\theta}(k)\right\|^2\Bigl|\mathcal{F}_{k-1}\right]\leq\left\|\tilde{\theta}(k-1)\right\|^2+\frac{M^2\left(m\sum_{p=0}^m\beta_p^2\right)}{k^2}.
		\end{aligned}
	\end{equation}
	Also
	\[
	\sum_{k=1}^{\infty}\frac{M^2\left(m\sum_{p=0}^m\beta_p^2\right)}{k^2}<\infty.
	\]
	Together with Lemma 1.2.2 in \cite{Chen}, it results in that $\tilde{\theta}(k)$ almost surely converges to a bounded limit. Notice that $\lim_{k\rightarrow\infty}\mathbb{E}\left[\left\|\tilde{\theta}(k)\right\|^2\right]=0$. Then, there exists a subsequence of $\tilde{\theta}(k)$ that almost surely converges to $0$. Consequently, $\tilde{\theta}(k)$ almost surely converges to $0$, i.e.,
	\[
	\lim_{k\rightarrow\infty}\hat{\theta}(k)=\theta,\ a.s.
	\]
	This completes the proof.
\end{proof}

Finally, we give the following result about the closed-loop system consisting of \eqref{multi-system} and \eqref{multi-adaptive-control}.

\begin{theorem}\label{multi-stable}
	Under Assumptions \ref{ass-reference-signal}-\ref{ass-multi-noise}, the adaptive control law \eqref{multi-adaptive-control} is asymptotically optimal, i.e.,
	\[
	\lim_{k\rightarrow\infty}\mathbb{E}\left[(y(k)-y^{*}(k))^2\right]=\mathbb{E}\left[w^2(1)\right].
	\]
\end{theorem}
\begin{proof}
	By Theorem \ref{multi-convergence}, we have $\hat{\theta}(k)\rightarrow\theta\ a.s.$ Then based on \eqref{multi-adaptive-control}, we obtain
%	\[
%	\phi^{\intercal}(k)\hat{\theta}(k-1)-y^{*}(k)\rightarrow 0,\ a.s.
%	\]
%	Hence,
	\[
	\begin{aligned}
		\phi^{\intercal}(k)\theta-y^{*}(k)=&\phi^{\intercal}(k)\left(\theta-\hat{\theta}(k-1)\right)+\phi^{\intercal}(k)\hat{\theta}(k-1)-y^{*}(k)\rightarrow 0,\ a.s.
	\end{aligned}
	\]
	According to Assumptions \ref{ass-reference-signal}, \ref{ass-multi-bounded} and Proposition \ref{pro-1}, we have
	\[
	\left\|\phi^{\intercal}(k)\theta-y^{*}(k)\right\|\leq \bar{M}M+\bar{y}^{*}.
	\]
	Also, by Lebesgue Dominated Convergence Theorem, we can obtain
	\[
	\lim_{k\rightarrow\infty}\mathbb{E}\left[(\phi^{\intercal}(k)\theta-y^{*}(k))^2\right]=0.
	\]
	Therefore, by \eqref{multi-system} and Assumption \ref{ass-noise}, we have
	\begin{equation*}
		\begin{aligned}
			\lim_{k\rightarrow\infty}\mathbb{E}\left[(y(k)-y^{*}(k))^{2}\right]=&\lim_{k\rightarrow\infty}\mathbb{E}\left[(\phi^{\intercal}(k)\theta-y^{*}(k)+w(k))^2\right]\\
			=&\lim_{k\rightarrow\infty}\mathbb{E}\left[(\phi^{\intercal}(k)\theta-y^{*}(k))^2\right]+\mathbb{E}\left[w^2(k)\right]\\
			=&\mathbb{E}\left[w^2(1)\right].
		\end{aligned}
	\end{equation*}
	This completes the proof.
\end{proof}
%%%%%%%%%%%%%%%%%%%%%%%%%%%%%%%%%%%%%%%%%%%%%%%%%%%%%%%%%%%%%%%%%%%%%%
\section{Simulations}
In this section, we give a simulation to illustrate the main results in this paper. 

\textbf{Example.} Consider the following stochastic regression model:
\begin{equation}\label{ex-model-1}
	\begin{aligned}
		y(k)=u(k)\theta(1)+u(k-1)\theta(2)+w(k),\	k=1,2,...\nonumber
	\end{aligned}
\end{equation}
where $\phi(k)\triangleq[u(k),u(k-1)]^{\intercal}\in\mathbb{R}^2$ and $\theta\triangleq[\theta(1),\theta(2)]^{\intercal}\in\mathbb{R}^2$ denote the system input and the unknown parameter, respectively, $u(z)=0$ for $z\leq0$, $\{w(k), k\geq 1\}$ be an i.i.d. sequence of standard normal random variables. The multi-threshold of sensor and the corresponding weight of the quantized outputs is $[-\infty,-2,0,2,\infty]$ and $[80,50,-50,-80]$. The reference signal $\{y^{*}(k)\}_{k\geq 1}$ is generated $y^*(2k)=2+e(k)$ and $y^*(2k-1)=1$, where $\{e(k)\}_{k\geq 1}$ is randomly selected in the interval $(0,0.1)$. Set the true parameter $\theta=[4,1]^{\intercal}$ be unknown but known as in ${\Omega}\triangleq\{[x,y]:\  |x|\leq 6, |y|\leq 2\}$ and $N=3$.

Let $\hat{\theta}(0)=[5,0]^{\intercal}$. Then we use the algorithm \eqref{identification}-\eqref{multi-adaptive-control} to estimate the unknown $\theta$ and track the reference signal $\{y^{*}(k)\}$. Figure 1 and Figure 2 show that estimation and tracking effects of the algorithm, respectively. From these, we learn that the estimation algorithm is convergent (see Figure 1(a)), the mean square
convergence speed can reach $O(\frac{1}{k})$ (see Figure 1(b)), and the adaptive tracking control law can also achieve asymptotically optimal tracking for the non-periodic reference signal (see Figure 2). These results are consistent with Theorems \ref{theorem-convergence-2} and \ref{multi-convergence}.
\begin{figure}[htbp]%
	\centering
	\subfigure[Convergence]{
		\includegraphics[
		height=2in,
		width=2.8in]%
		{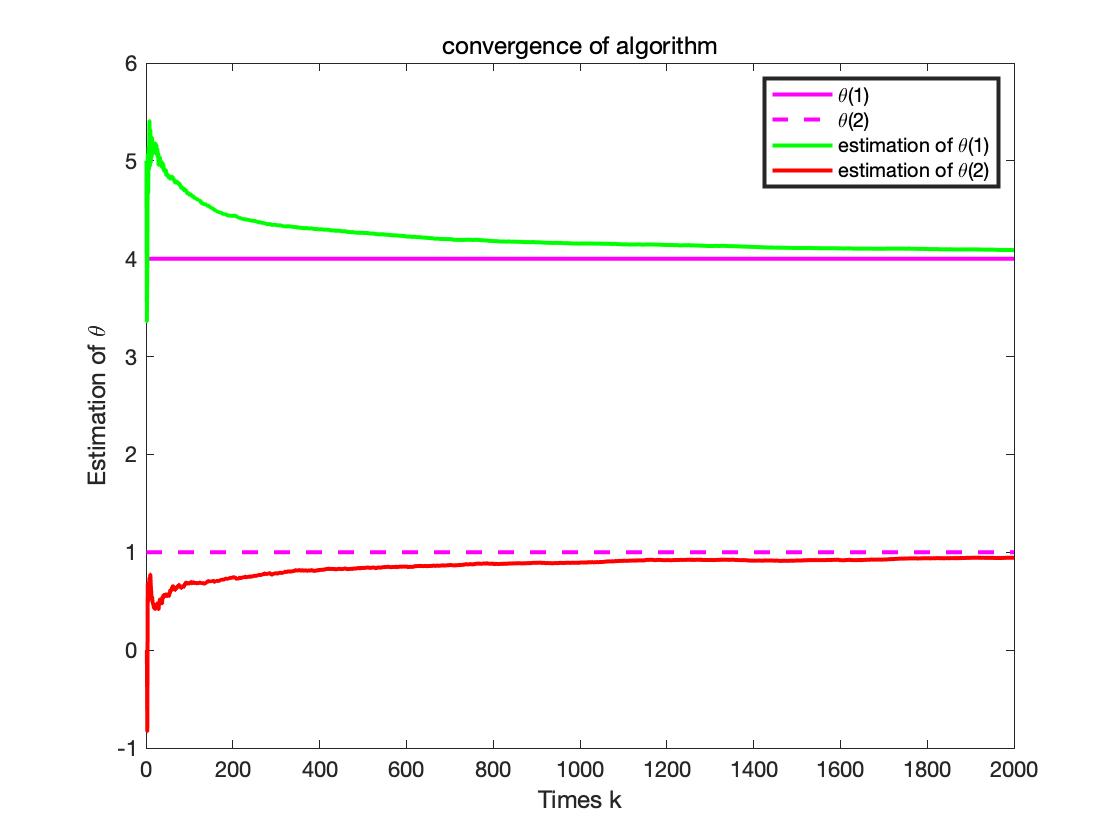}}
	\subfigure[Convergence speed]{
		\includegraphics[
		height=2in,
		width=2.8in]%
		{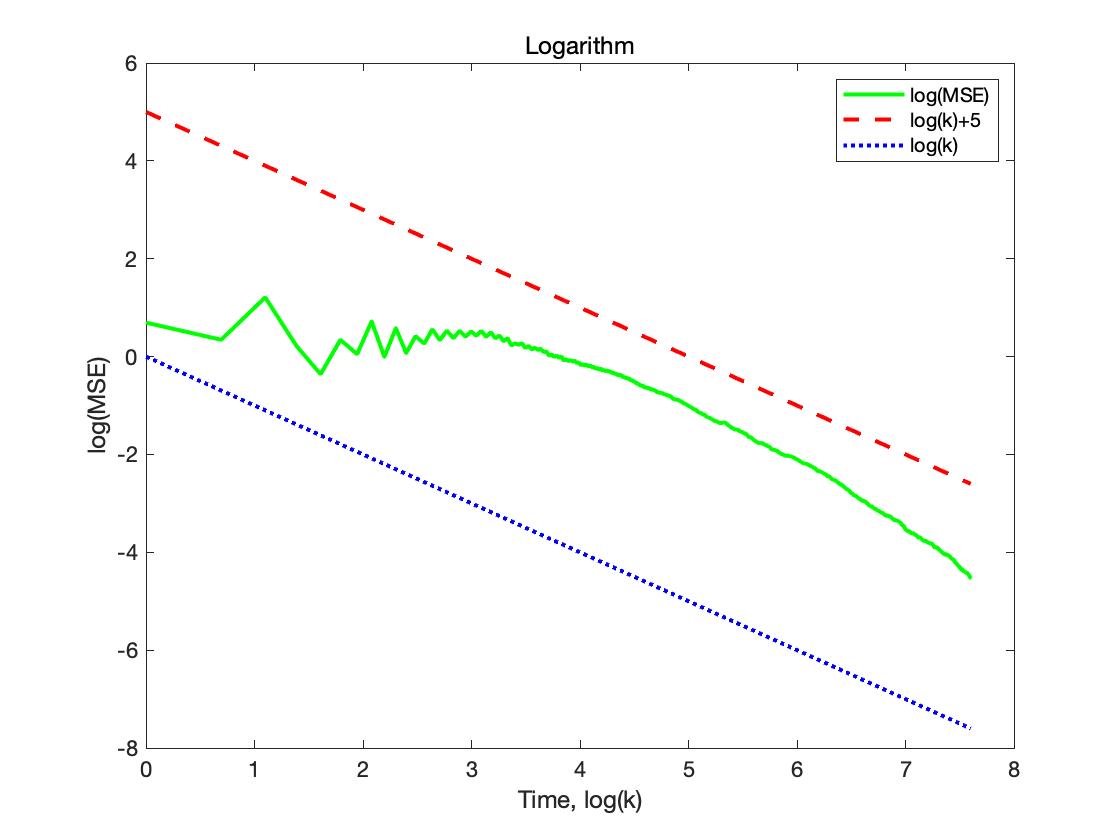}}
	\caption{ The performance of estimations}%
	%\label{absoluteerrorRK2}
\end{figure}
\begin{figure}[htbp]%
	\centering
	\subfigure[The trajectory of system output and reference signal at time 2k-1]{
		\includegraphics[
		height=2in,
		width=2.8in]%
		{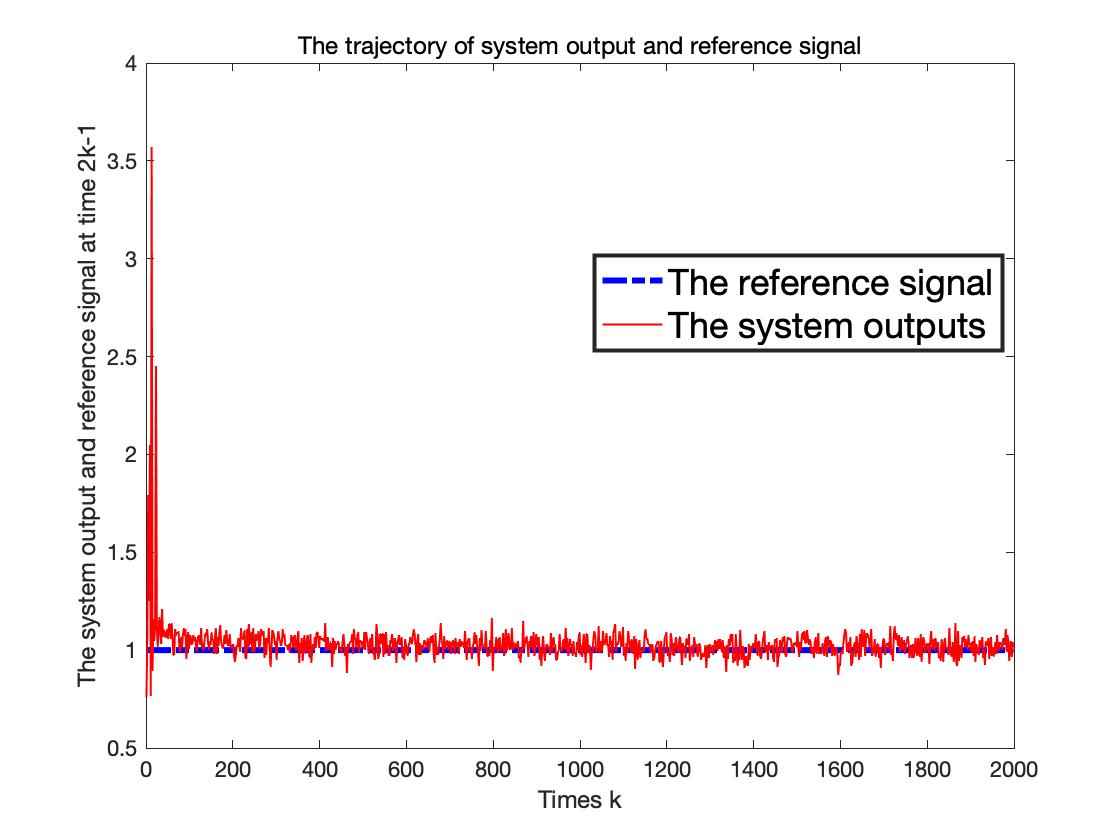}}
	\subfigure[The trajectory of system output and reference signal at time 2k]{
		\includegraphics[
		height=2in,
		width=2.8in]%
		{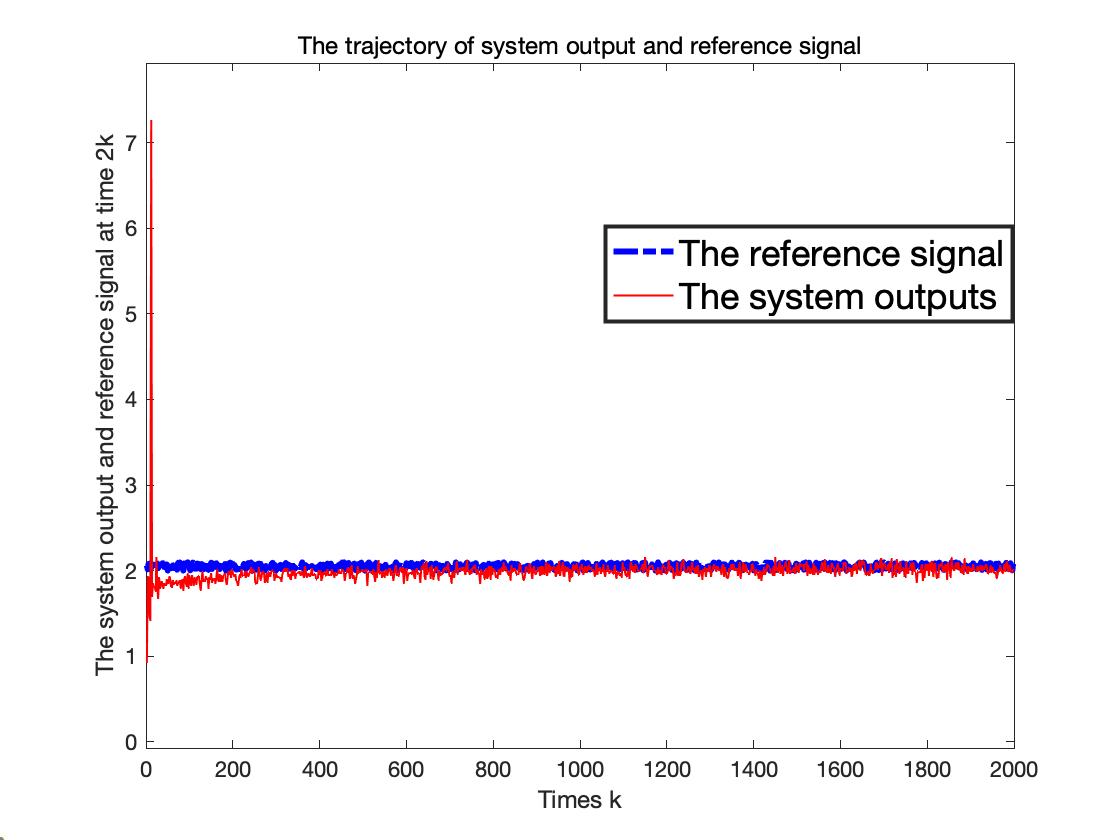}}
	%\subfigure[The tracking error]{
	%%	\includegraphics[
	%	height=2in,
	%	width=2in]%
	%	{2-Tracking-error.jpg}}
	\caption{Tracking effects}%
	%\label{absoluteerrorRK2}
\end{figure}

\section{Conclusion}
The motivation of this paper is to solve the adaptive tracking control problem under quantized observations for non-periodic reference signals. This paper presents an adaptive tracking control scheme that incorporates an online stochastic approximation-type estimation algorithm, and importantly this result is free from the reliance on periodic reference signals in the existing literature. The methodological innovation of this paper lies in overcoming the difficulty of ensuring that the adaptive tracking control designed by the certainty equivalent principle satisfies uniformly bounded and persistent excitation conditions by introducing two backward-shifted polynomials with time-varying parameters and a special projection structure. Finally, we obtain some results including the convergence speed of the estimation algorithm and asymptotically optimality of adaptive tracking control law for the non-periodic reference signals.

In the future, our work will focused on the following two areas: i) how to ensure that the designed adaptive tracking control satisfies the excitation condition without relying on the projection technique in this paper to achieve asymptotically optimal tracking of the non-periodic reference signal; ii) for more general high-order stochastic regression system containing quantized observations, how to extend the adaptive control scheme proposed in this paper to solve the tracking control problem for non-periodic reference signals.

\end{document}